\documentclass{llncs}

\usepackage{amssymb}
\usepackage{textcomp}
\usepackage{comment}
\usepackage{color}
\usepackage{graphicx}
\usepackage{amsmath}

\usepackage{algorithm}
\usepackage{algpseudocode}

\newcommand{\ket}[1]{|#1\rangle}

\begin{document}

\title{Quantum Algorithm for Dynamic Programming Approach for DAGs and Applications}
\author{Kamil~Khadiev$^{1}$ \and Liliya~Safina$^{1}$}

\institute{
Institute of Computational Mathematics and Information Technologies, Kazan Federal University, Kremlevskaya ul. 18, Kazan, Tatarstan,
420008 Russia
                \\ \email{kamilhadi@gmail.com}
}

\maketitle

\begin{abstract}In this paper, we present a quantum algorithm for the dynamic programming approach for problems on directed acyclic graphs (DAGs). The running time of the algorithm is $O(\sqrt{\hat{n}m}\log \hat{n})$, and the running time of the best known deterministic algorithm is $O(n+m)$, where $n$ is the number of vertices, $\hat{n}$ is the number of vertices with at least one outgoing edge; $m$ is the number of edges. We show that we can solve problems that use OR, AND, NAND, MAX, and MIN functions as the main transition steps. The approach is useful for a couple of problems. One of them is computing a Boolean formula that is represented by Zhegalkin polynomial, a Boolean circuit with shared input and non-constant depth evaluation. Another two are the single source longest paths search for weighted DAGs and the diameter search problem for unweighted DAGs.  
\end{abstract}

\section{Introduction}
\emph{Quantum computing} \cite{nc2010,a2017,aazksw2019part1} is one of the hot topics in computer science of the last decades.
There are many problems where quantum algorithms outperform the best-known classical algorithms. Some examples of such algorithms can be found here \cite{dw2001,quantumzoo,kiv2022,kb2022,ke2022,kr2021a,kr2021b,kk2021,abikkpssv2020,kksk2020,ki2019,kks2019}. 
The superiority of quantum over classical was shown for different computational models like query model, streaming processing models, communication models and others \cite{an2009,av2009,agky16,aakk2018,aakv2018,kkm2018,kk2017,ikpy2018,ikpy2021,l2009,kkzmkry2022,kk2022,kk2019,kk2020,kk2019disj,kkm2018}.

In this paper, we present the quantum algorithm for the class of problems on directed acyclic graphs (DAGs) that uses a dynamic programming approach.
The dynamic programming approach is one of the most useful ways to solve problems in computer science \cite{cormen2001}. The main idea of the method is to solve a problem using pre-computed solutions of the same problem, but with smaller parameters. Examples of such problems for DAGs that are considered in this paper are the single source longest path search problem for weighted DAGs and the diameter search problem for unweighted DAGs.

 Another example is a Boolean circuit with non-constant depth and shared input evaluation. A Boolean circuit can be represented as a DAG with conjunction (AND) or disjunction (OR) in vertices, and inversion (NOT) on edges. We present it as an algorithm for computing a Zhegalkin polynomial \cite{z1927,gg85,y2003}. The Zhegalkin polynomial is a way of a Boolean formula representation using ``exclusive or'' (XOR, $\oplus$), conjunction (AND), and the constants $0$ and $1$. 

The best-known deterministic algorithm for dynamic programming on DAGs uses the depth-first search algorithm (DFS) as a subroutine \cite{cormen2001}. Thus, this algorithm has at least the depth-first search algorithm's running time, that is $O(n+m)$, where $m$ is the number of edges and $n$ is the number of vertices. The query complexity of the algorithm is at least $O(m)$.

We suggest a quantum algorithm with the running time $O(\sqrt{\hat{n}m}\log \hat{n})$, where $\hat{n}$ is the number of vertices with a non-zero outgoing degree. In the case of $\hat{n}(\log \hat{n})^2<m$, it shows speed-up compared with a deterministic algorithm. The quantum algorithm can solve problems that use a dynamic programming algorithm with OR, AND, NAND, MAX, or MIN functions as transition steps. We use Grover's search  \cite{g96,bbht98} and D{\"u}rr and H{\o}yer maximum search \cite{dh96} algorithms to speed up our search. A similar approach has been applied by D{\"u}rr et al. \cite{dhhm2004,dhhm2006}; Ambainis and {\v{S}}palek \cite{as2006}; D{\"o}rn \cite{d2009,d2008phd} to several graph problems.

We apply this approach to four problems discussed above. 
The first of them involves computing Boolean circuits. Such circuits can be represented as AND-OR-NOT DAGs. Sinks of a such graph are associated with Boolean variables, and other vertices are associated with conjunction (AND) or a disjunction (OR); edges can be associated with the inversion (NOT) function. Quantum algorithms for computing  AND-OR-NOT trees were considered by Ambainis et al. \cite{avrz2010,a2007,a2010}. Authors present an algorithm with running time $O(\sqrt{N})$, where $N$ is the number of a tree's vertices. Other algorithms allow us to construct AND-OR-NOT DAGs of constant depth, but not a tree \cite{bkt2018,ckk2012}.

Our algorithm works with $O(\sqrt{\hat{n}m}\log \hat{n})$ running time for DAGs that can have non-constant depth.

It is known that any Boolean function can be represented as a Zhegalkin polynomial \cite{z1927,gg85,y2003}. The computation of a Zhegalkin polynomial is the second problem.  Such a formula can be represented as an AND-OR-NOT DAG. Suppose an original Zhegalkin polynomial has $t_c$ conjunctions and $t_x$ exclusive-or operations. Then the corresponding AND-OR-NOT DAG have $\hat{n}=O(t_x)$ non-sink vertices and $m=O(t_c+t_x)$ edges. 

If we consider AND-OR-NOT trees representation of the formula, then it has an exponential number of vertices $N\geq 2^{O(t_x)}$. The quantum algorithm for trees \cite{avrz2010,a2007,a2010} works in  $O(\sqrt{N})=2^{O(t_x)}$ running time. Additionally, the DAG that corresponds to a Zhegalkin polynomial has non-constant depth. Therefore, we cannot use algorithms from \cite{bkt2018,ckk2012} that work for circuits with shared input.

The second problem is the single source longest path search problem for a weighted DAG. The best deterministic algorithm for this problem works in $O(n+m)$ running time \cite{cormen2001}. In the case of a general graph (not a DAG), it is an NP-complete problem.  Our algorithm for DAGs works in $O(\sqrt{mn}\log n)$ running time. The third problem is the diameter search problem for an unweighted DAG. The best deterministic algorithms for this problem works in $O(n(n+m))$ expected running time \cite{cormen2001}. Our algorithm for DAGs works in expected $O(n(n+\sqrt{nm})\log n)$ running time. 
The fourth problem is the single source shortest path search problem for a weighted DAG. The best deterministic algorithms for this problem work in $O(n+m)$ running time \cite{cormen2001}. Our algorithm for DAGs works in expected $O(\sqrt{nm}\log n)$ running time. 

The paper is an extended version of \cite{ks2019} that was presented at UCNC2019 conference.

  The paper is organized as follows. We present definitions in Section \ref{sec:prlmrs}. Section  \ref{sec:algorithm} contains a general description of the algorithm. The application to an AND-OR-NOT DAG evaluation and Zhegalkin polynomial evaluation is in Section \ref{sec:andor-dag}. Section  \ref{sec:diameter} contains a solution for the single source longest path search problem for a weighted DAG and the diameter search problem for an unweighted DAG.

\section{Preliminaries}\label{sec:prlmrs}
Let us present definitions and notations that we use.

A graph $G$ is a pair $G=(V,E)$ where $V=\{v_1,\dots,v_n\}$ is a set of vertices, and $E=\{e_1,\dots,e_m\}$ is a set of edges, an edge $e\in E$ is a pair of vertices $e=(v,u)$, for $u,v\in V$.

A graph $G$ is directed if all edges $e=(v,u)$ are ordered pairs, and there are no bidirectional edges. Formally, for any $(v,u)\in E$, we have $(u,v)\not\in E$. In that case, an edge $e$ leads from vertex $v$ to vertex $u$. A graph $G$ is acyclic if no path starts and finishes in the same vertex.
We consider only directed acyclic graphs (DAGs) in the paper.

Let $D_i=(v:\exists e=(v_i,v)\in E)$ be a list of $v_i$ vertex's out-neighbors. Let $d_i=|D_i|$ be the out-degree of the vertex $v_i$.
 Let $L$ be a list of indices of sinks. Formally, $L = (i : d_i =0, 1 \leq i \leq n)$. Let $|L|$ be the length of the list, and $\hat{n}=n-|L|$. Let $Sources=\{j:d'_j=0\}$ be the list of source vertexes that have $0$ in-degree. 
 
Let $D'_i=(v:\exists e=(v,v_i)\in E)$ be a list of vertex whose in-neighbor is $v'$.  Let $d'_i=|D'_i|$ be the in-degree of the vertex $v_i$. 

 Let a DAG be {\em ordered} if it has two additional properties:
 \begin{itemize}
 \item topological sorted:  if there is an edge $e=(v_i,v_j)\in E$, then $i<j$;
 \item last $|L|$ vertices belong to $L$, formally $d_i=0$, for $i>\hat{n}$.
 \end{itemize}
 
 We use an {\em adjacency list model} as a model for graph representation. 
 The input is specified by $n$ arrays $D_i$, for $i\in\{1\dots n\}$. Here $D_i=(v:\exists e=(v_i,v)\in E)$ is a list of vertices that can be reached from a vertex $v_i$ in one step. Let $d_i=|D_i|$ be the degree of the vertex $v_i$.
\subsection{Quantum query model}
For our algorithms, we use some quantum subroutines, and the rest part of the algorithm is classical.   
One of the most popular computational models for quantum algorithms is the query model.
We use the standard form of the quantum query model. 
Let $f:D\rightarrow \{0,1\},D\subseteq \{0,1\}^M$ be an $M$ variable function. We wish to compute it on an input $x\in D$. We are given oracle access to the input $x$, i.e. it is realized by a specific unitary transformation usually defined as $\ket{i}\ket{z}\ket{w}\rightarrow \ket{i}\ket{z+x_i\pmod{2}}\ket{w}$ where the $\ket{i}$ register indicates the index of the variable we are querying, $\ket{z}$ is the output register, and $\ket{w}$ is some auxiliary work-space. An algorithm in the query model consists of alternating applications of arbitrary unitaries independent of the input and the query unitary, and measurement in the end. The smallest number of queries for an algorithm that outputs $f(x)$ with probability $\geq \frac{2}{3}$ on all $x$ is called the quantum query complexity of the function $f$ and is denoted by $Q(f)$. We use the running time term instead of query complexity.

We refer the readers to \cite{nc2010,a2017,aazksw2019part1} for more details on quantum computing.

\section{Quantum Depth-first Search Algorithm and Topological Sort Algorithm}
Here we discuss the modification of the quantum version of  Topological Sort and its base which is the Depth-first Search (DFS) Algorithm. The quantum version of the DFS algorithm was presented in \cite{f2008}. Its running time is $O(\sqrt{nm\log n})$. We present the modification that has $O(\sqrt{\hat{n}m}\log \hat{n})$. Here $m=|E|, n=|V|, \hat{n}=n-|L|$, $L$ is the set of sinks. Such complexity is important for our goals.

Let us present the DFS algorithm in Algorithm \ref{alg:dfs}. We assume that we have an array $visited$ such that $visited[i]=\textsc{True}$ if we have visited a vertex $v_i$; and $visited[i]=\textsc{False}$ otherwise. For the list $D_i=(v_{j_1},\dots,v_{j_{d_v}})$ of neighbors for a vertex $v_i$, let us have a  $\textsc{FirstOneSearch}(D_i,start)$ procedure that returns the minimal index $z$ such $z>start$ and $visited[v_{j_z}]=\textsc{False}$; or returns $\textsc{NULL}$ if there is no such an index. The procedure has running time $O(\sqrt{z-start}\log \hat{n})$ and error probability at most $\frac{1}{\hat{n}^2}$. It is a $O(\log \hat{n})$ times repetition of the quantum algorithm of searching minimal element satisfying a condition \cite{k2014,ll2015,ll2016,kkmsy2022}. The algorithm is based on Grover's Search Algorithm \cite{g96,bbht98}.

\begin{algorithm}
\caption{Quantum Depth-first Search Algorithm for DAGs. That is $\textsc{dfs}(i)$ procedure for processing a vertex $v_i$. 
}\label{alg:dfs}
\begin{algorithmic}
\State $visited[i]=\textsc{True}$
\If{$d_i>0$}\Comment{The vertex is not a sink}
\State $z\gets\textsc{FirstOneSearch}(D_i,0)$.\Comment{Here $(v_{j_1},\dots,v_{j_{d_i}})=D_i$}
\While{$z\neq \textsc{NULL}$}
\State $\textsc{dfs}(j_z)$
\State $z\gets\textsc{FirstOneSearch}(D_i,z)$.
\EndWhile
\EndIf
\end{algorithmic}
\end{algorithm}

The Topological sort algorithm is presented in Algorithms  \ref{alg:topsort} and \ref{alg:topsort-dfs}. It is based on the classical version of the algorithm \cite{cormen2001}. Let us have a list $order$. The $order$ list of result order. Assume that we have a procedure $\textsc{AddToTop}(order, i)$ that adds an element $i$ to the top of the list $order$.
Note, that the procedure can be implemented with $O(1)$ running time if we implement lists as the Linked List data structure \cite{cormen2001}.

For the list $D_i=(v_{j_1},\dots,v_{j_{d_v}})$ of neighbors for a vertex $v_i$, let the $\textsc{FirstOneSearch'}(D_i,start)$ procedure returns the minimal index $z$ such $z>start$, $visited[v_{j_z}]=\textsc{False}$ and $d_{j_z}>0$; or returns $\textsc{NULL}$ if there is not such an index. In other words, the procedure cannot return a sink vertex.

\begin{algorithm}
\caption{The base procedure of the quantum Topological sort algorithm is a modification of the DFS algorithm. It is $\textsc{TopSortDFS}(i)$ procedure for processing a vertex $v_i$. }\label{alg:topsort-dfs}
\begin{algorithmic}
\State $visited[i]=\textsc{True}$
\If{$d_i>0$}\Comment{The vertex is not a sink}
\State $z\gets\textsc{FirstOneSearch'}(D_i,0)$.\Comment{Here $(v_{j_1},\dots,v_{j_{d_i}})=D_i$}
\While{$z\neq \textsc{NULL}$}
\State $\textsc{dfs}(j_z)$
\State $z\gets\textsc{FirstOneSearch'}(D_i,z)$.
\EndWhile
\State $\textsc{AddToTop}(order,i)$
\EndIf
\end{algorithmic}
\end{algorithm}

\begin{algorithm}
\caption{The procedure of the quantum Topological sort algorithm for a graph $G$. It is $\textsc{TopSort}(G)$ procedure. }\label{alg:topsort}
\begin{algorithmic}
\State $order=()$\Comment{Initially, $order$ is an empty list}
\For{$i \in Sources$}
\If{$visited[i]=\textsc{False}$}
\State $\textsc{TopSortDFS}(i)$
\EndIf
\EndFor
\end{algorithmic}
\end{algorithm}

As a result of running the topological sort algorithm, we have an ordered DAG as a result of the new enumeration. Let $order=(i_1,\dots,i_{\hat{n}})$ and $L=(j_1,\dots,j_{|L|})$, where $n=\hat{n}+|L|$. The $r$-th vertex in the new enumeration is $v_{i_r}$ if $r\leq \hat{n}$; and $v_{j_{(r-\hat{n})}}$ if $r>\hat{n}$. For simplicity after a topological sort, we meet $r$-th vertex in the new enumeration, when writing $v_r$. 

Let us discuss the complexity of the Topological sort algorithm.
\begin{lemma}\label{lm:topsort-compl}
Algorithm \ref{alg:topsort} works with running time $O(\sqrt{\hat{n}m}\log \hat{n})$ and error probability at most $0.1$.
\end{lemma}
\begin{proof}
Let us consider processing $i$-th vertex by procedure $\textsc{TopSortDFS}$. Let $k_i$ be the number of steps of the while loop or the number of invocations of the $\textsc{FirstOneSearch}'$ procedure or the number of vertices that are reached from the $i$-th vertex.  Let $w_{i,1},\dots,w_{i,k_i}$ be the indices of these vertices in the list $D_i$. Assume that $w_{i,1}<\dots<w_{i,k_i}$

So, the running time of the procedure $\textsc{TopSortDFS}$ for  $i$-th vertex is
\[O(\sqrt{w_{i,1}}\log\hat{n})+O(\sqrt{w_{i,2}-w_{i,1}}\log\hat{n})+O(\sqrt{w_{i,3}-w_{i,2}}\log\hat{n})+\dots+O(\sqrt{w_{i,k_i}-w_{i,k_i-1}}\log\hat{n})\leq\]
\[O\left(\log \hat{n}\sqrt{k_i(w_{i,1}+w_{i,2}-w_{i,1}+w_{i,3}-w_{i,2}+\dots+w_{i,k_i}-w_{i,k_i-1})}\right)=\]
by the Cauchy--Bunyakovsky--Schwarz inequality
\[O(\log\hat{n}\sqrt{k_i\cdot w_{i,k_i}})\leq O(\log\hat{n}\sqrt{k_i\cdot d_i})\]

We invoke the $\textsc{TopSortDFS}$ procedure for all non-sink vertices ones. So, the total running time is by the Cauchy--Bunyakovsky--Schwarz inequality
\[\sum_{v_i\not\in L}O(\log\hat{n}\sqrt{k_i\cdot d_i})\leq O\left(\log\hat{n}\sqrt{\left(\sum_{v_i\not\in L} k_i\right)\cdot\left(\sum_{v_i\not\in L} d_i\right)}\right)=O(\log\hat{n} \sqrt{\hat{n}m})\]
The last equality is correct because $\textsc{FirstOneSearch}'$ returns only non-sink not visited vertices. So, therefore each non-sink vertex was returned at most once. Therefore, $\sum\limits_{v_i\not\in L} k_i\leq \hat{n}$. Additionally, $\sum\limits_{v_i\not\in L} d_i$ is at most the number of edges.

Each $\textsc{FirstOneSearch}'$ has error at most $\frac{1}{\hat{n}^2}$. We have at most $\hat{n}$ invocation of the procedure and the same number of independent error events. Therefore, the total error is at most $\frac{1}{\hat{n}}$.
\end{proof}
\section{Quantum Dynamic Programming Algorithm for DAGs}\label{sec:algorithm}

Let us describe an algorithm in the general case. 

Let us consider some problem $P$ on a directed acyclic graph $G=(V,E)$.
Suppose that we have a dynamic programming algorithm for $P$ or we can say that there is a solution of the problem $P$ that is equivalent to computing a function $f$ for each vertex.
As a function $f$ we consider only functions from a set ${\cal F}$ with the following properties:
\begin{itemize}
\item $f:V\to \Sigma$.
\item The result set $\Sigma$ can be the set of real numbers $\mathbb{R}$, or integers $\{0,\dots,{\cal Z}\}$, for some integer ${\cal Z}>0$.
\item if $d_i>0$ then $f(v_i)=h_i(f(u_1),\dots,f(u_{d_i}))$, where functions $h_i$ are such that $h_i:\Sigma^{d_i}\to \Sigma$; $(u_1,\dots,u_{d_i})=D_i$.
\item if $d_i=0$ then  $f(v_i)$ is classically computable in constant time. 
\end{itemize}

Suppose there is a quantum algorithm $Q_i$ that computes the function $h_i$ with running time $T(k)$, where $k$ is the length of the argument for the function $h_i$. Then we can suggest the procedure that is presented in Algorithm \ref{alg:general-dp}. We assume that we have a $\textsc{TopSort}(G)$ procedure that is an implementation of the topological sort algorithm from the previous section. After applying the topological sort algorithm we have an ordered graph. If the graph is already ordered, then we should not invoke the topological sort algorithm.

\begin{algorithm}
\caption{Quantum Algorithm for Dynamic programming approach on DAGs.
Let $t=(t[1],\dots,t[\hat{n}])$ be an array which stores results of the function $f$. Let $t_f(j)$ be a function such that $t_f(j)=t[j]$, if $j\leq \hat{n}$; $t_f(j)=f(j)$, if $j> \hat{n}$. Note that $j> \hat{n}$ means $v_j\in L$.
}\label{alg:general-dp}
\begin{algorithmic}
\State \textsc{TopSort}(G)\Comment{We do it only if $G$ is not ordered}
\For{$i=\hat{n}\dots 1$}
  \State $t[i] \gets Q_i(t_f(j_1),\dots,t_f(j_{d_i}))$, where $(v_{j_1},\dots,v_{j_{d_i}})=D_i$
\EndFor
\State \Return $t[1]$
\end{algorithmic}
\end{algorithm}

Let us discuss the running time of Algorithm \ref{alg:general-dp}. The proof is simple, but we present it for completeness. 
\begin{lemma}\label{lm:dp-time}
Suppose that the quantum algorithm $Q_i$ works in $T_i(k)$ running time, where $k$ is a length
of an argument, where $i\in\{1,\dots,n\}$. Then Algorithm \ref{alg:general-dp} works in $T^1 = \sum\limits_{i\in\{1,\dots,n\}\backslash L}
T_i(d_i)$ running time if the graph is ordered. If the graph is not ordered, then the running time is $T^1+O(\sqrt{\hat{n}m}\log\hat{n})$.
\end{lemma}
\begin{proof}
If the graph is not ordered and we invoke the $\textsc{TopSort}(G)$ procedure, then it has $O(\sqrt{nm\log \hat{n}})$ running time due to Lemma \ref{lm:topsort-compl}.

Note, that when we compute $t[i]$, we already have computed $t_f(j_1),\dots,t_f(j_{d_i})$ or we can compute them in constant running time because for all $e=(v_i,v_j)\in E$ we have $i<j$. 

The complexity of processing vertex $v_i$ is $T_i(d_i)$, where $i \in \{1, \dots , n\}\backslash L$.

The algorithm processes vertices one by one. Therefore 
\[T^1 = \sum\limits_{i\in\{1,\dots,n\}\backslash L}
T_i(d_i).\]

The total complexity in the case of not ordered graph is $T^1+O(\sqrt{\hat{n}m}\log n) $
\end{proof}

Note, that quantum algorithms have a probabilistic behavior. Let us compute the error probability for
Algorithm \ref{alg:general-dp}.

\begin{lemma}\label{lm:dp-err}
Suppose the quantum algorithm $Q_i$ for the function $h_i$ has the error probability $\varepsilon(n)$, where $i\in\{1,\dots,n\}\backslash L$. Then the error probability of Algorithm \ref{alg:general-dp} is at most $1 - (1-\frac{1}{\hat{n}})(1 - \varepsilon(n))^{\hat{n}}$.
\end{lemma}
\begin{proof}
Let us compute the success probability for Algorithm \ref{alg:general-dp}. Suppose that all vertices are
computed with no error. The probability of this event is $(1 - \varepsilon(n))^{\hat{n}}$ because an error of each invocation is an independent event. Additionally, we have an error of $\textsc{TopSort}$ procedure that is at most $\frac{1}{\hat{n}}$ due to Lemma \ref{lm:topsort-compl}. 

Therefore, the error probability for Algorithm \ref{alg:general-dp} is at most $1-0.1(1-\varepsilon(n))^{\hat{n}}$, for $\hat{n}=n-|L|$. 
\end{proof}

For some functions and algorithms, we do not have a requirement that all arguments of $h$ should be computed with no error. In that case, we will get a better error probability. This situation is discussed in Lemma \ref{lm:dp-max-compl}.

\subsection{Functions for Vertices Processing}
We can choose the following functions as a function $h$
\begin{itemize}
\item Conjunction ($AND$ function). For computing this function, we can use Grover's search
algorithm \cite{g96,bbht98} for searching $0$ among arguments. If the element that we found is $0$, then the result is $0$. If the element is $1$, then there are no $0$s, and the result is $1$.
\item Disjunction ($OR$ function). We can use the same approach, but here we search $1$s.
\item Sheffer stroke (Not $AND$ or $NAND$ function). We can use the same approach as for AND function, but here we search $1$s. If we
found $0$ then the result is $1$; and $0$, otherwise.

\item Minimum function ($MIN$). We can use the D{\"u}rr and H{\o}yer minimum search algorithm \cite{dh96,dhhm2004}.
\item Maximum function ($MAX$). We can use the same algorithm as for the minimum.
\item Other functions that have quantum algorithms.
\end{itemize}

As we discussed before, $AND$, $OR$, and $NAND$ functions can be computed using the Grover search algorithm. Therefore algorithm for these functions on vertex $v_i$ has an error $\varepsilon_i\leq 0.5$ and running time is $T(d_i)=O(\sqrt{d_i})$, for  $i\in \{1,\dots,n\}\backslash L$. These results follow from \cite{bhmt2002,gr2005,bbht98,g96}. We have a similar situation for computing maximum and minimum functions \cite{dh96,dhhm2004}.

If we use these algorithms in Algorithm \ref{alg:general-dp} then we obtain the error probability  $1-(0.5)^{\hat{n}+1}$ due to Lemma \ref{lm:dp-err}.

At the same time, the error is one-sided. That is why we can apply the boosting technique to reduce the error probability. 

Suppose, we have a quantum algorithm $Q$ and a classical algorithm $A$ for a function $h\in\{MAX,MIN, AND, OR, NAND\}$. Let $k$ be the number of algorithm's invoking according to the boosting technique for reducing the error probability. The number $k$ is integer and $k\geq 1$. Let $(x_1,\dots,x_d)$ be arguments (input data) of size $d$.
Let us denote it as $\hat{Q}^k(x_1,\dots,x_d)$. It returns the result of $A\left(Q(x_1,\dots,x_d),\dots,Q(x_1,\dots,x_d)\right)$. Here $A$ has $k$ arguments. Suppose, we have a temporary array $b=(b[1],\dots,b[k])$ where we store results of each invocation $b[i] = Q(x_1,\dots,x_d)$ for $i\in\{1,\dots,k\}$. Then, we compute $A(b[1],\dots,b[k])$.

If we analyze the algorithm, then we can see that it has the following property:
\begin{lemma} \label{lm:boosting}
 Let $(x_1,\dots,x_d)$ be an argument (input data) of size $d$, for a function $h(x_1,\dots,x_d)\in\{MAX,MIN, OR, AND, NAND\}$. Let $k$ be a number of algorithm's invokations. The number $k$ is integer and $k\geq 1$.   Then the expected running time of the boosted version $\hat{Q}^k(x_1,\dots,x_d)$ of the quantum algorithm $Q(x_1,\dots,x_d)$ is $O\left(k\sqrt{d}\right)$ and the error probability is at most $\frac{1}{2^k}$.
\end{lemma}
\begin{proof}
Due to \cite{dh96,dhhm2004}, the expected running time of the algorithm $Q$ is $O\left(\sqrt{d}\right)$, and the error probability is at most $0.5$.  We apply $A$  for $k$ copies of $Q(x_1,\dots,x_d)$. That works in $O(k\cdot \sqrt{d})$. We have an error if all $k$ copies have an error. So, the probability of error is at most $\frac{1}{2^k}$.
\end{proof}

Let us apply the previous two lemmas to Algorithm \ref{alg:general-dp} and functions from the set  $\{AND,OR,NAND,MAX,MIN\}$.

\begin{lemma}\label{lm:dp-general-complx}

Suppose that a problem $P$ on a DAG $G=(V,E)$ has a dynamic programming algorithm such that functions  $h_i\in\{AND,OR,NAND,MAX,MIN\}$, for $i\in\{1,\dots,\hat{n}\}$. Then there is a quantum dynamic programming algorithm $A$ for the problem $P$ that has running time $O(\sqrt{\hat{n}m}\log \hat{n})=O(\sqrt{nm}\log n)$ and error probability $O(1/\hat{n})$. Here $m=|E|, n=|V|, \hat{n}=n-|L|$, $L$ is the set of sinks. 
\end{lemma}
\begin{proof}

Let us choose $k=2\log_2 \hat{n}$ in Lemmas \ref{lm:boosting}. Then the error probabilities for the algorithms $Q^{2\log_2 \hat{n}}$ are $O\left(0.5^{2\log_2 \hat{n}}\right)=O\left(1/\hat{n}^2\right)$. The  running time is $O(\sqrt{d_i}\log \hat{n})$.

Due to Lemma \ref{lm:dp-err}, the probability of error is at most $\varepsilon(\hat{n})= 1-\left(1-\frac{1}{\hat{n}^2}\right)^{\hat{n}}$. Note that 
\[\lim\limits_{\hat{n}\to \infty}\frac{\varepsilon(\hat{n})}{1/\hat{n}}= \lim\limits_{\hat{n}\to \infty} \frac{1-\left(1-\frac{1}{\hat{n}^2}\right)^{\hat{n}}}{1/\hat{n}}=1;\]
 Hence, $\varepsilon(\hat{n})=O(1/\hat{n})$.

Due to Lemma  \ref{lm:dp-time}, the running time is
\[T^1=\sum\limits_{i\in\{1,\dots,n\}\backslash L}
T_i(d_i) \leq \sum\limits_{i\in\{1,\dots,n\}\backslash L}
O\left(\sqrt{d_i}\log \hat{n}\right)  =O\left((\log_2 \hat{n})\cdot\sum\limits_{i\in\{1,\dots,n\}\backslash L}
\sqrt{d_i}\right).\]
Due to the Cauchy-Bunyakovsky-Schwarz inequality, we have 
\[\sum\limits_{i\in\{1,\dots,n\}\backslash L} \sqrt{d_i}\leq \sqrt{\hat{n}\sum\limits_{i\in\{1,\dots,n\}\backslash L}d_i}\]
Note that $d_i=0$, for $i\in L$. Therefore, $\sum\limits_{i\in\{1,\dots,n\}\backslash L}d_i=\sum\limits_{i\in\{1,\dots,n\}}d_i=m$, because $m=|E|$ is the total number of edges. Hence,
\[\sqrt{\hat{n}\sum\limits_{i\in\{1,\dots,n\}\backslash L}d_i}=\sqrt{\hat{n}\sum\limits_{i\in\{1,\dots,n\}}d_i}=\sqrt{\hat{n}m}.\] Therefore, $T^1\leq O(\sqrt{\hat{n}m}\log \hat{n})= O(\sqrt{nm}\log n)$. 

Additionally, we have the running time of $\textsc{TopSort}$ procedure that is also $O(\sqrt{\hat{n}m}\log \hat{n})$. The total running time is $O(\sqrt{\hat{n}m}\log \hat{n})$.
\end{proof}

If $h_i\in\{MAX,MIN\}$ and a graph is ordered, then we can do a better estimation of the running time.
\begin{lemma}\label{lm:dp-max-compl}
Suppose that a problem $P$ on an ordered DAG $G=(V,E)$ has a dynamic programming algorithm such that functions  $h_i\in\{MAX, MIN\}$, for $i\in\{1,\dots,\hat{n}\}$ and the solution is $f(v_a)$ for some $v_a\in V$. Then there is a quantum dynamic programming algorithm $A$ for the problem $P$ that has expected running time $O(\sqrt{\hat{n}m}\log q)=O(\sqrt{nm}\log n)$ and error probability $O(1/q)$, where $q$ is the length of the path to the farthest vertex from the vertex $v_a$. Here $m=|E|, n=|V|, \hat{n}=n-|L|$, $L$ is the set of sinks.
\end{lemma}
\begin{proof}
Let $Q$ be the D{\"u}rr-H{\o}yer quantum algorithm for $MAX$ or $MIN$ function. Let $\hat{Q}^q$ be the boosted version of $Q$. Let us analyze the algorithm.

Let us consider a vertex $v_i$ for $i\in\{1,\dots,n\}\backslash L$. When we process $v_i$, we should compute $MAX$ or $MIN$ among $t_f(j_1),\dots,t_f(j_{d_i})$. Without loss of generality, we can say that we compute the MAX function. Let $r$ be an index of the maximal element. It is required to have no error for computing $t[j_r]$. At the same time, if we have an error on processing $v_{j_w}$, $w\in\{1,\dots,d_i\}\backslash\{r\}$; then we get a value $t[j_w]<f(v_{j_w})$. In that case, we still have $t[j_r]>t[j_w]$. Therefore, an error can be on processing of any vertex $v_{j_w}$.

Let us focus on the vertex $v_a$. For computing $f(v_a)$ with no error, we should compute $f(v_{a_1})$ with no error. Here $v_{a_1}\in D_a$ such that maximum is reached on $v_{a_1}$. For computing $f(v_{a_1})$ with no error, we should compute $f(v_{a_2})$ with no error. Here $v_{a_2}\in D_{a_1}$ such that maximum is reached on $v_{a_2}$ and so on. Hence, for solving the problem with no error, we should process only at most  $q$ vertices with no error.  

Therefore, the probability of error for the algorithm is
\[1-\left(1-\left(\frac{1}{2}\right)^{2\log q}\right)^{q}=O\left(\frac{1}{q}\right)
\]
because $\lim\limits_{q\to\infty}\frac{1-\left(1-\frac{1}{q^2}\right)^{q}}{1/q}=1$.
\end{proof}
\section{Quantum Algorithms for Evolution of Boolean Circuits with Shared Inputs and Zhegalkin Polynomial}\label{sec:andor-dag}
Let us apply ideas of quantum dynamic programming algorithms on DAGs to AND-OR-NOT DAGs. 

It is known that any Boolean function can be represented as a Boolean circuit with AND, OR, and NOT gates \cite{y2003,ab2009}. Any such circuit can be represented as a DAG with the following properties:
\begin{itemize}
\item sinks are labeled with variables. We call these vertices ``variable-vertices''.
\item There are no vertices $v_i$ such that $d_i=1$.
\item If a vertex $v_i$ such that $d_i\geq 2$; then the vertex labeled with Conjunction or Disjunction.  We call these vertices ``function-vertices''.
\item Any edge is labeled with $0$ or $1$.
\item There is one particular root vertex $v_s$.
\end{itemize}

The graph represents a Boolean function that can be evaluated in the following way. 
We associate a value $r_i\in\{0,1\}$ with a vertex $v_i$, for $i\in\{1,\dots,n\}$. If $v_i$ is a variable-vertex, then $r_i$ is a value of a corresponding variable. If $v_i$ is a function-vertex labeled by a function $h_i\in\{AND, OR\}$, then $r_i=h_i\left(r_{j_1}^{\sigma(i,j_1)},\dots,r_{j_{w}}^{\sigma(i,j_w)}\right)$, where $w=d_i$, $(v_{j_1},\dots,v_{j_w})=D_i$, $\sigma(i,j)$ is a label of an edge $e=(i,j)$. Here,  we say that $x^1=x$ and $x^0=\neg x$ for any Boolean variable $x$. The result of the evaluation is $r_s$.

An AND-OR-NOT DAG can be evaluated using the following algorithm that is a modification of Algorithm \ref{alg:general-dp}:

\begin{algorithm}
\caption{Quantum Algorithm for AND-OR-NOT DAGs evaluation.
Let $r=(r_1,\dots,r_n)$ be an array that stores the results of functions $h_i$. Let a variable-vertex $v_i$ be labeled by $x(v_i)$, for all $i\in L$. Let $Q_i$ be a quantum algorithm for $h_i$; and $\hat{Q}_i^{2\log_2 \hat{n}}$ be a boosted version of $Q_i$ (Lemma \ref{lm:boosting}).  Let $t_f(j)$ be a function such that $t_f(j)=r_j$, if $j\leq \hat{n}$; $t_f(j)=x(v_j)$, if $j> \hat{n}$. 
}\label{alg:and-or}
\begin{algorithmic}
\State $\textsc{TopSort}(G)$\Comment{If the graph is not ordered}
\For{$i=\hat{n}\dots s$}
  \State $t[i] \gets Q_i^{2\log_2 \hat{n}}(t_f(j_1)^{\sigma(i,j_1)},\dots,t_f(j_w)^{\sigma(i,j_w)})$, where $w=d_i$, $(v_{j_1},\dots,v_{j_{w}})=D_i$.
\EndFor
\State \Return $t[s]$
\end{algorithmic}
\end{algorithm}

Algorithm \ref{alg:and-or} has the following property:
\begin{theorem}\label{th:and-or-dag}
Algorithm \ref{alg:and-or} evaluates an AND-OR-NOT DAG $G=(V,E)$ with running time $O(\sqrt{\hat{n}m}\log \hat{n})=O(\sqrt{nm}\log n)$ and error probability $O(1/\hat{n})$. Here $m=|E|, n=|V|, \hat{n}=n-|L|$, $L$ is the set of sinks. 
\end{theorem}
\begin{proof}
Algorithm \ref{alg:and-or} evaluates the AND-OR-NOT DAG $G$ by the definition of AND-OR-NOT DAG for the Boolean function $F$.
Algorithm \ref{alg:and-or} is almost the same as Algorithm \ref{alg:general-dp}. The difference is labels of edges. At the same time, the Oracle gets information on an edge in constant time. Therefore, the running time and the error probability of $Q_i^{2\log_2 \hat{n}}$ does not change. Hence, using the proof similar to the proof of Lemma \ref{lm:boosting}, we obtain the claimed running time and error probability.
\end{proof}

Another way of a Boolean function representation is a NAND DAG or Boolean circuit with NAND gates. \cite{y2003,ab2009}.
We can represent a NAND formula as a DAG with similar properties as AND-OR-NOT DAG, but function-vertices have only NAND labels. 
At the same time, if we want to use more operations, then we can consider NAND-NOT DAGs and NAND-AND-OR-NOT DAGs:
\begin{theorem}\label{th:nand-dag}
Algorithm \ref{alg:and-or} evaluates a NAND-AND-OR-NOT DAG and a NAND-NOT DAG. If we consider a DAG $G=(V,E)$, then these algorithms work with running time $O(\sqrt{\hat{n}m}\log \hat{n})=O(\sqrt{nm}\log n)$ and error probability $O(1/\hat{n})$. Here $m=|E|, n=|V|, \hat{n}=n-|L|$, $L$ is the set of sinks. 
\end{theorem}
\begin{proof}
The proof is similar to the proofs of Lemma \ref{lm:dp-general-complx} and Theorem \ref{th:and-or-dag}.
\end{proof}

Theorems \ref{th:and-or-dag} and \ref{th:nand-dag} give us quantum algorithms for Boolean circuits with shared input and non-constant depth. At the same time, existing algorithms \cite{avrz2010,a2007,a2010,bkt2018,ckk2012} are not applicable in the case of shared input and non-constant depth.

The third way of representation of a Boolean function is the Zhegalkin polynomial which is representation using $AND, XOR$ functions and the $0,1$ constants \cite{z1927,gg85,y2003}: for some integers $k,t_1,\dots,t_k$,
\[F(x)=ZP(x)=a\oplus\bigoplus_{i=1}^{k}C_i\mbox{, where }a\in\{0,1\}, C_i=\bigwedge_{z=1}^{t_i} x_{j_z}\]

At the same time, it can be represented as an AND-OR-NOT DAG with a logarithmic depth and shared input or an AND-OR-NOT tree with an exponential number of vertices and separated input. That is why the existing algorithms from \cite{avrz2010,a2007,a2010,bkt2018,ckk2012} cannot be used or work in exponential running time. 

\begin{theorem}\label{th:xor-dag}
Algorithm \ref{alg:and-or} evaluates the XOR-AND DAG $G=(V,E)$ with running time $O(\sqrt{\hat{n}m}\log \hat{n})=O(\sqrt{nm}\log n)$ and error probability $O(1/\hat{n})$. Here $m=|E|, n=|V|, \hat{n}=n-|L|$, $L$ is the set of sinks.
\end{theorem}
\begin{proof}
 $XOR$ operation is replaced by two $AND$, one $OR$ vertex and $6$ edges because for any Boolean $a$ and $b$ we have $a \oplus b = a \wedge \neg b \vee \neg a \wedge b$. So, we can represent the original DAG as an AND-OR-NOT DAG using $\hat{n}'\leq 3\cdot \hat{n} =O(\hat{n})$ vertices. The number of edges is $m'\leq 6\cdot m=O(m)$. Due to Theorem \ref{th:and-or-dag}, we can construct a quantum algorithm with running time $O(\sqrt{\hat{n}m}\log \hat{n})$ and error probability $O(1/\hat{n})$.
\end{proof}

The previous theorem shows us the existence of a quantum algorithm for Boolean circuits with XOR, NAND, AND, OR, and NOT gates. Let us present the result for the Zhegalkin polynomial.

\begin{corollary}\label{cr:zp}
Suppose that Boolean function $F(x)$ can be represented as Zhegalkin polynomial for some integers $k,t_1,\dots,t_k$: $F(x)=ZP(x)=a\oplus\bigoplus_{i=1}^{k}C_i$, where $a\in\{0,1\},$ $C_i=\bigwedge_{z=1}^{t_i} x_{j_z}$. Then, there is a quantum algorithm for $F$ with running time $O\left(\sqrt{k(k+t_1+\dots+t_k)}\log k\right)$ and error probability $O(1/k)$.
\end{corollary}
\begin{proof}
Let us present $C_i$ as one $AND$ vertex with $t_i$ outgoing edges. $XOR$ operation is replaced by two $AND$, one $OR$ vertex and $6$ edges. So, $m=6\cdot (k-1) + t_1 +\dots + t_k=O(k+t_1+\dots+t_k)$, $\hat{n}=3\cdot (k-1) + k=O(k)$.
Due to Theorem \ref{th:xor-dag}, we obtain the claimed properties.
\end{proof}
\section{The Quantum Algorithm for the Single Source Longest and Shortest Paths Problems for a Weighted DAG and the Diameter Search Problem for Unweighted DAG}\label{sec:diameter}

In this section, we consider two problems for DAGs.
\subsection{The Single Source Longest Path Problem for Weighted DAG}
Let us apply the approach to the Single Source Longest Path problem.

Let us consider a weighted DAG $G=(V,E)$ and the weight of an edge $e=(v_i,v_j)$ is $w(i,j)$, for $i,j\in \{1,\dots,n\}, e\in E$.

Let us have a vertex $v_s$, and we should compute $t[1],\dots,t[n]$. Here $t[i]$ is the length of the longest path from $v_s$ to $v_i$. If a vertex $v_i$ is not reachable from $v_s$ then $t[i]=-\infty$.

Let us present the algorithm for the longest paths lengths computing. 

\begin{algorithm}
\caption{Quantum Algorithm for the Single Source Longest Path Search problem.
Let $t=(t[1],\dots,t[n])$ be an array that stores results for vertices. Let $Q$ be the D{\"u}rr-H{\o}yer quantum algorithm for $MAX$ function. Let $\hat{Q}^{2\log_2 (n)}$ be a boosted version of $Q$  (Lemma \ref{lm:boosting}).  
}\label{alg:longest-path}
\begin{algorithmic}
\State $\textsc{TopSort}(G)$
\State $t\gets (-\infty,\dots,-\infty)$
\State $t[s]\gets 0$
\For{$i=s+1\dots n$}
  \State $t[i] \gets \hat{Q}^{2\log_2 n}(t[j_1]+w(i,j_1),\dots,t[j_w]+w(i,j_w))$, where $w=|D'_i|$, $(v_{j_1},\dots,v_{j_{w}})=D'_i$.
\EndFor
\State \Return $t$
\end{algorithmic}
\end{algorithm}

Algorithm \ref{alg:longest-path} has the following property:
\begin{theorem}\label{th:longest-path}
Algorithm \ref{alg:longest-path} solves the Single Source Longest Path Search problem with expected running time $O(\sqrt{nm}\log n)$ and error probability $O(1/n)$.
\end{theorem}
\begin{proof}
Let us prove the correctness of the algorithm. In fact, the algorithm computes $t[i]=\max(t[j_1]+w(i,j_1),\dots,t[j_w]+w(i,j_w))$. Assume that $t[i]$ is less than the length of the longest path. Then there is $v_z\in D_i'$ that precedes $v_i$ in the longest path. Therefore, the length of the longest path is $t[z]+w(i,z)>t[i]$. This claim contradicts the definition of $t[i]$ as the maximum. The bounds for the running time and the error probability follow from Lemmas \ref{lm:dp-time}, \ref{lm:dp-err}. 
\end{proof}

\subsection{The Single Source Shortest Path Problem for Weighted DAG}
Let us apply the approach to the Single Source Shortest Path problem.

Let us consider a weighted DAG $G=(V,E)$ and the weight of an edge $e=(v_i,v_j)$ is $w(i,j)$, for $i,j\in \{1,\dots,n\}, e\in E$.

Suppose we have a vertex $v_s$, and we should compute $t[1],\dots,t[n]$. Here $t[i]$ is the length of the shortest path from $v_s$ to $v_i$. If a vertex $v_i$ is not reachable from $v_s$ then $t[i]=+\infty$.

Let us present the algorithm for the shortest paths lengths computing. 

\begin{algorithm}
\caption{Quantum Algorithm for the Single Source Shortest Path Search problem.
Let $t=(t[1],\dots,t[n])$ be an array that stores results for vertices. Let $Q$ be the D{\"u}rr-H{\o}yer quantum algorithm for $MIN$ function. Let $\hat{Q}^{2\log_2 (n)}$ be a boosted version of $Q$  (Lemma \ref{lm:boosting}).  
}\label{alg:shortest-path}
\begin{algorithmic}
\State $\textsc{TopSort}(G)$
\State $t\gets (+\infty,\dots,+\infty)$
\State $t[s]\gets 0$
\For{$i=s+1\dots n$}
  \State $t[i] \gets \hat{Q}^{2\log_2 n}(t[j_1]+w(i,j_1),\dots,t[j_w]+w(i,j_w))$, where $w=|D'_i|$, $(v_{j_1},\dots,v_{j_{w}})=D'_i$.
\EndFor
\State \Return $t$
\end{algorithmic}
\end{algorithm}

Algorithm \ref{alg:shortest-path} has the following property:
\begin{theorem}\label{th:shortest-path}
Algorithm \ref{alg:shortest-path} solves the Single Source Shortest Path Search problem with expected running time $O(\sqrt{nm}\log n)$ and error probability $O(1/n)$.
\end{theorem}
\begin{proof}
Let us prove the correctness of the algorithm. In fact, the algorithm computes $t[i]=\max(t[j_1]+w(i,j_1),\dots,t[j_w]+w(i,j_w))$. Assume that $t[i]$ is less than the length of the longest path. Then there is $v_z\in D_i'$ that precedes $v_i$ in the longest path. Therefore, the length of the longest path is $t[z]+w(i,z)>t[i]$. This claim contradicts the definition of $t[i]$ as the maximum. The bounds for the running time and the error probability follow from Lemmas \ref{lm:dp-time}, \ref{lm:dp-err}. 
\end{proof}

\subsection{The Diameter Search Problem for an Unweighted DAG}

Let us consider an unweighted DAG $G=(V,E)$. Let $len(i,j)$ be the length of the shortest path between $v_i$ and $v_j$. If the path does not exist, then $len(i,j)=-1$.  The diameter of the graph $G$ is $diam(G)=\max\limits_{i,j\in\{1,\dots,|V|\}}len(i,j)$.  
For a given graph $G=(V,E)$, we should find the diameter of the graph.

It is easy to see that the diameter is the length of a path between a non-sink vertex and some other vertex. If this fact is false, then the diameter is $0$. 

Using this fact, we can present the algorithm. 

\begin{algorithm}
\caption{Quantum Algorithm for the Diameter Search problem.
Let $t^z=(t^z[1],\dots,t^z[n])$ be an array that stores the shortest paths from vertices to vertex $v_z\in V\backslash L$. Let $Q$ be the D{\"u}rr-H{\o}yer quantum algorithm for the $MIN$ function. Let $\hat{Q}^{2\log_2 (n)}$ be a boosted version of $Q$  (Lemma \ref{lm:boosting}). 
}\label{alg:diam}
\begin{algorithmic}
\State $max\_len\gets -\infty$
\For{$z=\hat{n}\dots 1$}
\State $t^z\gets (+\infty,\dots,+\infty)$
\State $t^z[z]\gets 0$
\For{$i=z+1\dots n$}
  \State $t^z[i] \gets \hat{Q}^{2\log_2 n}(t^{j_1}[i],\dots,t^{j_w}[i])+1$, where $w=|D'_i|$, $(v_{j_1},\dots,v_{j_{w}})=D'_i$.
  \If{$t^z[i]>max\_len$}
  \State $max\_len\gets t^z[i]$
  \EndIf
\EndFor
\EndFor
\State $diam(G)=max\_len$
\State \Return $diam(G)$
\end{algorithmic}
\end{algorithm}
    
Algorithm \ref{alg:diam} has the following property:
\begin{theorem}\label{th:diam}
Algorithm \ref{alg:diam} solves the Diameter Search problem with expected running time $O(\hat{n}(n+\sqrt{nm})\log n)$ and error probability $O(1/n)$.
\end{theorem}
\begin{proof}
The correctness of the algorithm can be proven similar to the proof of Theorem \ref{th:longest-path}. The bounds for the running time and the error probability follow from Lemmas \ref{lm:dp-time}, \ref{lm:dp-err}. 
\end{proof}

\paragraph*{Acknowledgements.} 
A part of the study was supported by Kazan Federal University for the state assignment in the sphere of
scientific activities, project No. 0671-2020-0065. 

We thank Andris Ambainis, Alexander Rivosh, and Aliya Khadieva for their help and useful discussions.

\end{document}